\documentclass[12pt,a4paper]{article}

\usepackage{amssymb}
\usepackage{amsmath,amsthm,amssymb,color}
\usepackage{graphicx}
\date{}
\usepackage{url}
\urldef{\mailsa}\path|{alfred.hofmann, ursula.barth, ingrid.haas, frank.holzwarth,|
\urldef{\mailsb}\path|anna.kramer, leonie.kunz, christine.reiss, nicole.sator,|
\urldef{\mailsc}\path|erika.siebert-cole, peter.strasser, lncs}@springer.com|

\newtheorem{lem}{Lemma}[section]
\newtheorem{thm}[lem]{Theorem}
\newtheorem{prop}[lem]{Proposition}
\newtheorem{coro}[lem]{Corollary}

\newtheorem{rmk}[lem]{Remark}

\newcommand{\uhom}[1]{\underline{\mathrm{Hom}}}

\newcommand{\arr}[2]{\begin{array}{#1}#2\end{array}}

\begin{document}

\title{On constructions and properties of $(n,m)$-functions with maximal number of bent components}

\author{Lijing Zheng~$^{\diamond}$, Jie Peng~$^{*}$, Haibin Kan~$^{\diamond }$, Yanjun Li$~^{*}$, and  Juan Luo
\thanks{ J. Peng, Y. Li, and J. Luo are with Mathematics and Science College of
Shanghai Normal University, Shanghai, China, 200234, (E-mails:~\{jpeng,luojuan\}@shnu.edu.cn, yanjlmath90@163.com).
\newline \indent $^{\diamond}$ L. Zheng and H. Kan are with the Department of Computer Sciences, Fudan University,
Shanghai, 200433, China; L. Zheng is also with School of Mathematics and Physics, University of South China, Hengyang, Hunan, 421001, China, (E-mails:~zhenglijing817@163.com, hbkan@fudan.edu.cn).}}

\maketitle
\begin{abstract}
For any positive integers $n=2k$ and $m$ such that $m\geq k$, in this paper we show the maximal number of bent components of any $(n,m)$-functions is equal to $2^{m}-2^{m-k}$, and for those attaining the equality, their algebraic degree is at most $k$. It is easily seen that all $(n,m)$-functions of the form $G(x)=(F(x),0)$ with $F(x)$ being any vectorial bent $(n,k)$-function, have the maximum number of bent components. Those simple functions $G$ are called trivial in this paper. We show that for a power $(n,n)$-function, it has such large number of bent components if and only if it is trivial under a mild condition. We also consider the $(n,n)$-function of the form $F^{i}(x)=x^{2^{i}}h({\rm Tr}^{n}_{e}(x))$, where $h: \mathbb{F}_{2^{e}} \rightarrow \mathbb{F}_{2^{e}}$, and show that $F^{i}$ has such large number if and only if $e=k$, and $h$ is a permutation over $\mathbb{F}_{2^{k}}$. It proves that all the previously  known nontrivial such functions are subclasses of the functions $F^{i}$.  Based on the Maiorana-McFarland class, we present constructions of large numbers of $(n,m)$-functions with maximal number of bent components for any integer $m$ in bivariate representation.  We also determine the differential spectrum and Walsh spectrum of the constructed functions. It is found that our constructions can also provide new plateaued vectorial functions.
\end{abstract}

\medskip


\section{Introduction}

Throughout this paper, we often identify the finite field $\mathbb{F}_{2^{n}}$ with $\mathbb{F}^{n}_{2}$, the $n$-dimensional vector space over $\mathbb{F}_{2}$. Any function $F: \mathbb{F}_{2^{n}}\rightarrow\mathbb{F}_{2^{m}}$ is called an {\it $(n,m)$-function}. When $m=1$ it is usually called a {\it Boolean function}.
{\it Bent functions}, as a special class of Boolean functions, were introduced by Rothaus \cite{Rothaus} in 1976. A bent function is a Boolean function in even number of variables which is maximally nonlinear in the sense that its Hamming distance to all affine functions is optimal. It corresponds to the fact that the {\it Walsh transform} of a bent function in $n$ variables takes precisely the values $\pm 2^{\frac{n}{2}}$. Over the last four decades, bent functions have attracted a lot of research interest because of their relations to coding theory and applications in cryptography. A survey on bent functions can be found in \cite{Car-Mes16} as well as the book \cite{Mes16}.

The bent property of Boolean functions has been extended to general $(n,m)$-functions by requesting that all the nonzero linear combinations of the {\it coordinate functions} of $F$ are bent functions. Such functions are called {\it vectorial bent}. The construction of vectorial bent functions has been initially considered by Nyberg \cite{Nyb}. It was shown in \cite{Nyb} that vectorial bent $(n,m)$-functions exist if and only if $n$ is even and $m\leq n/2$. Such functions can be constructed by using some known classes of bent functions, such as the {\it Maiorana-McFarland class}~(denoted by $\mathcal{M}$) and the Dillon's {\it partial spread class} (denoted by $\mathcal{PS}$).  In the finite field version, the bent property of $(n,m)$-functions $F: \mathbb{F}_{2^{n}} \rightarrow \mathbb{F}_{2^{m}}$ with $\mathbb{F}_{2^{m}}\subseteq \mathbb{F}_{2^{n}}$ can be established directly as it was done in \cite{MPS} for the trace functions ${\rm Tr}^{n}_{k}(\sum\limits^{2^{k}}_{i=0}\alpha_{i}x^{i(2^{k}-1)})$, where $n=2k$ and $\alpha_{i}\in \mathbb{F}_{2^{n}}$, see \cite{Mes16} for more constructions of vectorial bent functions.

An interesting and important problem concerning bent functions and vectorial bent functions is that, how large could the number of {\it bent components} of an $(n,m)$-function be ($n$ being even)? Thanks to Nyberg's result, the largest value of this number is $2^m-1$ when $m\leq n/2$. While for $m>n/2$, the answer is unknown before. Very recently, Pott et al. studied this number for the case $m=n$. In \cite{Pott} they showed that this number is upper bounded by $2^{n}-2^{n/2}$, which is attained by the functions of the form $F(x)=x^{2^{i}}(x+x^{2^{k}})$, where $n=2k, 0 \leq i\leq n-1$. They also pointed out that vectorial
functions with maximal number of bent components can produce new vectorial bent functions. More precisely, for any $(n,n)$-function $G$, if ${\rm Tr}^{n}_{1}(\alpha G(x))$ is bent for any $\alpha\in\mathbb{F}_{2^{n}}\setminus\mathbb{F}_{2^{k}}$, then ${\rm Tr}^{n}_{k}(\alpha G(x))$ is a vectorial bent function for any $\alpha \in \mathbb{F}_{2^{n}}\setminus\mathbb{F}_{2^{k}}$. Therefore, ${\rm Tr}^{n}_{k}(\gamma F(x))$ is vectorial bent for any  primitive element $\gamma$ in $\mathbb{F}^{\ast}_{2^{n}}$ and the functions $F$ listed above. Later, Mesnager et al. presented a class of $(n,n)$-functions with maximal number of bent components of the form $H(x)=x^{2^{i}}(x+x^{2^{k}}+\sum\limits^{\rho}_{j=1}\gamma^{(j)}(x^{2^{t_{j}}}+x^{2^{t_{j}+k}}))$ under some additional conditions, where $\gamma^{(j)}\in \mathbb{F}_{2^{k}}$ \cite{MZTZ}. In the same paper, the authors also proved that APN plateaued functions cannot have the maximum number of bent components. In some earlier works, the number of bent components of APN functions and differentially 4-uniform functions were discussed and some lower bounds were obtained (see \cite{Bcc,Cp}).

Let $n=2k$, and $m\geq k$. In this paper, at first we show that the maximal number of bent components of any $(n,m)$-function is equal to $2^{m}-2^{m-k}$, and when the equality holds, the \emph{algebraic degree} of such functions is at most $k$. For any vectorial bent $(n,k)$-function $F(x)$, with the inclusion map $i: \mathbb{F}^{k}_{2}\rightarrow \mathbb{F}^{m}_{2}$, $a\mapsto (a,0)$, it can also be regarded as an $(n,m)$-function which has the maximal number $2^{m}-2^{m-k}$ of bent components. For this reason, we call the $(n,m)$-function $G(x)=(F(x),0)$ \emph{trivial}. Our main contributions to the constructions and characterizations of such functions are summarized as follows:

$\bullet$  The case of $(n,m)$-functions.

Based on the  Maiorana-McFarland class,  we present a nontrivial construction of such functions in {\it bivariate representation} $G(x,y)=(x^{2^{i}}\pi(y)+g(y),h(y))$, where $(x,y)\in \mathbb{F}_{2^{k}}\times \mathbb{F}_{2^{k}}$, $\pi: \mathbb{F}_{2^{k}}\rightarrow \mathbb{F}_{2^{k}}$ is a permutation, $g(y):\mathbb{F}_{2^{k}}\rightarrow \mathbb{F}_{2^{k}}$ and $h(y):\mathbb{F}_{2^{k}}\rightarrow \mathbb{F}_{2^{m-k}}$ are arbitrary functions.

$\bullet$ The case of $(n,n)$-functions.

1) Let $F(x)=x^d$. We show that  $F(x)$ has the maximal number of bent components if and only if $F(x)$ is a vectorial bent $(n,k)$-function. Say $F$ is trivial.

2) For any positive integer $e$ dividing $k$, we consider the $(n,n)$-functions in {\it univariate representation} of the form
                                               $$F^{i}(x)=x^{2^{i}}h({\rm Tr}^{n}_{e}(x)),$$
where $0\leq i\leq n-1$, $h: \mathbb{F}_{2^{e}}\rightarrow \mathbb{F}_{2^{e}}$ is an arbitrary function. We show that the functions ${\rm Tr}^{n}_{1}(\alpha F^{i}(x))$ is bent for any $\alpha \in \mathbb{F}_{2^{n}}\setminus\mathbb{F}_{2^{k}}$ if and only if $e=k$, and $h$ is a permutation over $\mathbb{F}_{2^{e}}$. Hence the functions $F^{i}(x)=x^{2^{i}}h({\rm Tr}^{n}_{k}(x))$ yield new vectorial bent functions ${\rm Tr}^{n}_{k}(\alpha F^{i}(x))$, where $h$ is a permutation over $\mathbb{F}_{2^{k}}$.

One will see the functions $F(x)=x^{2^{i}}(x+x^{2^{k}})$ given by \cite{Pott} and $H(x)=x^{2^{i}}(x+x^{2^{k}}+\sum\limits^{\rho}_{j=1}\gamma^{(j)}(x^{2^{t_{j}}}+x^{2^{t_{j}+k}}))$ given by \cite{MZTZ} are subclasses of our functions $F^{i}(x)$, which means that all the known nontrivial $(n, n)$-functions with maximum number of bent components are covered by our results, to the best of our knowledge.

The rest of the paper is organized as follows. Some basic definitions are given in Section 2. In Section 3, we prove that the maximum number of bent components of $(n=2k,m)$-functions is equal to $2^m-2^{m-k}$, and when the equality holds, the algebraic degree of such functions is at most $k$. We give a characterization of the power $(n,n)$-functions admitting such large number of bent components. In Section 4,  we put forward a construction of $(n,m)$-functions in bivariate representation with maximum number of bent components, and determine the {\it differential spectrum} and {\it Walsh spectrum}. We also show that the $(n,n)$-functions $F^{i}(x)=x^{2^{i}}h({\rm Tr}^{n}_{e}(x))$ has the maximal number of bent components if and only if $e=k$, and $h$ is a permutation over $\mathbb{F}_{2^{e}}$, where $n=2k$, and $e$ is a positive integer dividing $k$, $h$ is an arbitrary map from $\mathbb{F}_{2^{e}}$ to $\mathbb{F}_{2^{e}}$. This yields new vectorial bent function ${\rm Tr}^{n}_{k}(\alpha x^{2^{i}}h({\rm Tr}^{n}_{k}(x)))$, where $h$ is a permutation over $\mathbb{F}_{2^{k}}$, $\alpha$ is a primitive element in $\mathbb{F}_{2^{n}}$. Moreover, we compare our new construction with previous ones, and prove that all the known nontrivial results are covered by ours. A few concluding remarks are given in Section 5.

\section{Preliminaries}

Let $\mathbb{F}_{2^{n}}$ be the finite field consisting of $2^{n}$ elements, then the group of units of $\mathbb{F}_{2^{n}}$, denoted by $\mathbb{F}^{\ast}_{2^{n}}$, is a cyclic group of order $2^{n}-1$. Throughout this paper we always identify $\mathbb{F}_{2^{n}}$ with the vector space $\mathbb{F}^{n}_{2}$ over $\mathbb{F}_{2}$. Any function $F: \mathbb{F}_{2^{n}}\rightarrow\mathbb{F}_{2^{m}}$ is called an $(n,m)$-function ({\it cryptographic function}, {\it vectorial Boolean function} or {\it substitution box}). Usually $(n,1)$-functions are called Boolean functions in $n$ variables, the set of which is denoted by $\mathcal{B}_{n}$. As cryptographic functions, the main cryptographic criteria for Boolean functions are algebraic degree, {\it nonlinearity} and {\it algebraic immunity}, etc. (see \cite{Carletb1}), and the main cryptographic criteria for vectorial Boolean functions are algebraic degree, nonlinearity and {\it differential uniformity}, etc. (see \cite{Carletb2}).

The trace function ${\rm Tr}^{n}_{m}: \mathbb{F}_{2^{n}}\rightarrow \mathbb{F}_{2^{m}}$, where $m|n$, is defined as
$${\rm Tr}^{n}_{m}(x)\!=\!x\!+\!x^{2^{m}}\!+\!x^{2^{2m}}\!+\!\cdots\!+\!x^{2^{(n/m\!-\!1)m}}, ~\forall ~x\in \mathbb{F}_{2^{n}}.$$ When $m=1$, it is also called the absolute trace. In this paper, $\langle,\rangle$ denotes the usual {\it inner product} in a vector space over $\mathbb{F}_{2}$. For any $\alpha=(\alpha_1,\ldots,\alpha_n), \beta=(\beta_1,\ldots,\beta_n)\in\mathbb{F}_2^{n}$, one has $\langle\alpha,\beta\rangle=\sum_{i=1}^n\alpha_i\beta_i$. While in the finite field $\mathbb{F}_{2^{n}}$, we take $\langle\alpha,\beta\rangle={\rm Tr}^{n}_{1}(\alpha\beta)$ for any $\alpha,\beta\in\mathbb{F}_{2^{n}}$.

 For any $(n,m)$-function $F=(f_1,\ldots,f_m)$, where $f_1,\ldots,f_m\in\mathcal{B}_{n}$, all the nonzero linear combinations of
 $f_i, 1\leq i\leq m$ are
called the components of $F$. When $F$ is viewed as a mapping from the finite field $\mathbb{F}_{2^{n}}$ to $\mathbb{F}_{2^{m}}$,
the components of $F$ can be represented as ${\rm Tr}_1^m(\lambda F(x)), ~\lambda\in\mathbb{F}_{2^{m}}^*.$

A Boolean function $f\in\mathcal{B}_{n}$ can
be uniquely represented by a multivariate polynomial as
\begin{eqnarray*}
f(x_1,\ldots,x_n)\!=\!\sum_{I\subseteq
\{1,2,\ldots,n\}}a_I\prod_{i\in I}x_i,
\end{eqnarray*} where the
number of variables in the highest order term with nonzero
coefficient is called the algebraic degree of $f$. While for a general $(n,m)$-function $F$, the highest algebraic degree of its coordinate functions is called the algebraic degree of $F$. The function $F$ is called {\it quadratic} if its algebraic degree is no more than 2.

The {\it Walsh transform} of a Boolean function $f\in \mathcal{B}_{n}$ at a point $\lambda \in \mathbb{F}_{2^{n}}$ is defined by
$$W_{f}(\lambda)\!=\!\sum\limits_{x\in \mathbb{F}_{2^{n}}}(\!-\!1)^{f(x)\!+\!{\rm Tr}^{n}_{1}(\lambda x)}.$$
 A function $f\in \mathcal{B}_{n}$ is called bent if $|W_{f}(\lambda)|=2^{\frac{n}{2}}$ for all $\lambda \in \mathbb{F}_{2^{n}}$. It is well known that bent functions exist if and only if $n$ is even. To find bent functions is a hard task. So far the known largest class of bent functions is the so called Maiorana-McFarland class of bent functions~(denoted by $\mathcal{M}$) \cite{McF}, which are of the form $f(x,y)={\rm Tr}^{k}_{1}(\pi(y)x)+g(y)$, $(x,y)\in \mathbb{F}_{2^{k}}\times  \mathbb{F}_{2^{k}}$, where $k$ is some positive integer, $\pi$ is a permutation over $\mathbb{F}_{2^{k}}$ and $g$ is an arbitrary Boolean function in $k$ variables. The function $f$ is called {\it plateaued} when either it is bent or $W_f$ takes three values
 $\{0,\pm 2^{s}\}$ for some integer $n/2< s\leq n$.

The nonlinearity of an $(n,m)$-function $F$ and hereby its resistance to {\it linear cryptanalysis} \cite{Matsui} is measured through the {\it extended Walsh spectrum}
$$\{\ast~ |W_{F}(\sigma,\gamma)|:\gamma \in \mathbb{F}^{m}_{2}\backslash\{0\},\sigma \in \mathbb{F}^{n}_{2} ~\ast\},$$
where
$$W_{F}(\sigma,\gamma)\!=\!\sum\limits_{x\in  \mathbb{F}^{n}_{2}}(\!-\!1)^{\langle \sigma,x\rangle\! +\!\langle \gamma,F(x)\rangle }.$$
The function $F$ is said to be a {\it vectorial~bent~function} of dimension $m$ if all the components of $F$ are bent. In other words, $F$ is vectorial bent if and only if $|W_{F}(\sigma,\gamma)|=2^{n/2}$, for any $\gamma \in \mathbb{F}^{m}_{2}\backslash\{0\}$ and for any $\sigma\in\mathbb{F}^{n}_{2}$. $F$ is said to be a {\it plateaued vectorial function} if all the components are plateaued Boolean functions.

For an $(n,m)$-function $F$, denote by $\delta_{F}(a,b)$ the number of solutions $x\in\mathbb{F}^{n}_{2}$ such that $F(x+a)+F(x)=b$, for any $a\in\mathbb{F}^{n}_{2}$ and $b\in\mathbb{F}^{m}_{2}$.
The multiset
$$\{\ast ~\delta_{F}(a,b):a\in \mathbb{F}^{n}_{2}\backslash \{0\}, b\in\mathbb{F}^{m}_{2}~ \ast \},$$
is called the {\it differential spectrum} of $F$, and $F$ is said to be differentially $\delta$-uniform, where $\delta=\max_{a\in \mathbb{F}^{n}_{2}\backslash \{0\}, b\in\mathbb{F}^{m}_{2}}\delta_{F}(a,b)$.
In particular, $F$ is called {\it almost perfect nonlinear}~(${\rm APN}$ for short)~if $\delta=2$. For cryptographic applications, it is usually required that the differential uniformity of $F$ should be as low as possible (see for instance \cite{Pt} and the references therein).

Two $(n,m)$-functions $F$ and $G$ are called {\it extended affine equivalent} (EA-equivalent) if there exist some affine permutation $L_1$ over $\mathbb{F}_{2^n}$ and some affine permutation $L_2$ over $\mathbb{F}_{2^m}$, and some affine function $A$ such that $F=L_2\circ G\circ L_1+A$. They are called {\it Carlet-Charpin-Zinoviev equivalent} (CCZ-equivalent) if there exists some affine automorphism $L=(L_1, L_2)$ of $\mathbb{F}_{2^n}\times \mathbb{F}_{2^m}$, where $L_1: \mathbb{F}_{2^n}\times \mathbb{F}_{2^m}\rightarrow \mathbb{F}_{2^n}$ and $L_2: \mathbb{F}_{2^n}\times \mathbb{F}_{2^m}\rightarrow \mathbb{F}_{2^m}$ are affine functions, such that $y=G(x)$ if and only if $L_2(x, y)=F\circ L_1(x, y)$. It is well known that EA-equivalence is a special kind of CCZ-equivalence, and that  CCZ-equivalence preserves the extended Walsh spectrum and the differential spectrum (but not for algebraic degree) \cite{CCZ}.

\section{$(n,m)$-functions with Maximal Number of Bent Components}

In this section, we will determine the tight upper bound for the number of bent components of general $(n,m)$-functions, where $n$ is even and $m\geq \frac{n}{2}$, and present a trivial construction of $(n,m)$-functions which reach this bound.

\begin{thm}\cite{Bose-Burton} \label{Bose-Burton} Let $\mathbb{F}_{q}$ be the finite field of order $q=p^s$ and $PG(n,q)$ the projective space of dimension $n$ over $\mathbb{F}_{q}$. A set $S$ of points in $PG(n,q)$ that meets all $(n-k)$-dimensional subspaces of $PG(n,q)$ has at least $\frac{q^{k+1}-1}{q-1}$ points with equality if and only if $S$ is a subspace of dimension $k$.
\end{thm}

As it has been pointed out by the authors of \cite{Pott}, one can get the following Lemma when applying Theorem \ref{Bose-Burton} to the case of $p=2$.

\begin{lem}\label{key}\cite[Corollary 1]{Pott}. A set $S$ of elements in $\mathbb{F}_{2}^{n}\backslash\{0\}$ meeting all $(n+1-k)$-dimensional subspaces of $\mathbb{F}_{2}^{n}$ has at least $2^k-1$ elements with equality if and only if $S\cup \{0\}$ is a $k$-dimensional subspace of $\mathbb{F}_{2}^{n}$.
\end{lem}

\begin{thm}\label{number} Let $n=2k$, $m$ be  positive integers with $m\geq k$, and $F$ an $(n,m)$-function. Denote by $S=\{v \in\mathbb{F}_{2}^{m}: x\mapsto \langle v,F(x)\rangle \text{is not bent}\}.$
Then $|S|\geq 2^{m-k}$, with equality if and only if $S$ is an $(m-k)$-dimensional subspace of $\mathbb{F}_{2}^{m}$. In particular, the maximal number of bent components of $F$ is $2^{m}-2^{m-k}$.
\end{thm}
\begin{proof}  With similar arguments as in \cite[Theorem 3]{Pott}, one can obtain the desired assertions. For the convenience of the readers, we give the proof here. If $|S|<2^{m-k}$, then there are at most $2^{m-k}-2$ nonzero elements $v$ for which $x\mapsto \langle v,F(x)\rangle$ is not bent. Due to Lemma \ref{key}, $S$ cannot meet all subspaces of dimension $m+1-(m-k)=k+1$, hence there exists some subspace $T$ of dimension  $k+1$ disjoint from $S\backslash \{0\}$. Then we have  $T\cap (S\backslash \{0\})=\emptyset$, and hence $T\subseteq \mathbb{F}_{2}^{m} \backslash  (S\backslash \{0\})$. Therefore $T\subseteq \{v \in\mathbb{F}_{2}^{m}: x\mapsto \langle v,F(x)\rangle \text{~is bent}\}\cup\{0\},$ that is, there is a vectorial bent function from $\mathbb{F}_{2}^{n}$ to $\mathbb{F}_{2}^{k+1}$, which is impossible due to Nyberg's theorem \cite{Nyb} on vectorial bent function which says that vectorial bent $(n,m)$-functions exist if and only if $n$ is even and $m\leq \frac{n}{2}$.
\end{proof}

It is well known that the algebraic degree of a bent function in $n$ variables is upper bounded by $n/2$. This upper bound is tight, and is valid for vectorial bent functions as well. In fact, it is also true for all $(n,m)$-functions with maximum number of bent components.

\begin{thm}\label{deg} Let $n=2k$, $m\geq k$ be positive integers. Let $F$ be an $(n,m)$-function with maximum number of bent components, then its algebraic degree is at most $k$.
\end{thm}
\begin{proof} For any non-bent component $f$ of $F$, let $g$ be a bent component, then by Theorem \ref{number} $f+g$ is bent. One then immediately gets the desired assertion by the upper bound for the algebraic degree of bent functions.
\end{proof}

For any vectorial bent $(n,k)$-function $F(x)$ ($n=2k$), with the inclusion map $i: \mathbb{F}^{k}_{2}\rightarrow \mathbb{F}^{m}_{2}$, $a\mapsto (a,0)$, it can also be regarded as an $(n,m)$-function which can be easily seen to have the maximal number $2^{m}-2^{m-k}$ of bent components. For this reason, in this paper we call such $(n,m)$-function $G(x)=(F(x),0)$ trivial.

As it has been pointed out by Pott et al., for an $(n,n)$-function $F$ with maximal number of bent components,  it can produce many vectorial bent $(n,k)$-functions \cite[Proposition 3]{Pott}. Explicitly, if ${\rm Tr}^{n}_{1}(\alpha F(x))$ is bent for any $\alpha \in \mathbb{F}_{2^{n}}\backslash \mathbb{F}_{2^{k}}$, then ${\rm Tr}^{n}_{k}(\alpha F(x))$ is a vectorial bent function for any $\alpha \in \mathbb{F}_{2^{n}}\backslash \mathbb{F}_{2^{k}}$. With this observation, we revisit the following $(n,n)$-function $F(x)$  which is essentially due to Leander and Kholosha \cite{Leander-Kholosha}.

\noindent {\bf Example 1.} Let $n=2k$, $F(x)=\sum\limits^{2^r-1}_{i=1}x^{(i2^{k-r}+1)(2^k-1)+1}$ with $1 \leq r\leq k$, and ${\rm gcd}(r,k)=1$. It has been observed by Li et al. that, for any $\alpha \in \mathbb{F}_{2^{n}}\backslash \mathbb{F}_{2^{k}}$, ${\rm Tr}^{n}_{1}(\alpha F(x))$ is bent \cite[Theorem 2]{Li2013}. Thus, by Theorem \ref{number}, $F$ has $2^{n}-2^k$ bent components, and  ${\rm Tr}^{n}_{1}(\alpha F(x))$ is bent if and only if $\alpha \in \mathbb{F}_{2^{n}}\backslash \mathbb{F}_{2^{k}}$. One will see that $F(x)$ is in fact a vectorial bent $(n,k)$-function. Indeed, for $1\leq i\leq 2^{r}-1$ denote $(i2^{k-r}+1)(2^k-1)+1$ by $d_{i}$, one has $F(x)=x^{(2^{k-1}+1)(2^k-1)+1}+\sum\limits^{2^{r-1}-1}_{i=1}(x^{d_{i}}+x^{d_{i}2^{k}})$. Note that $2(1+(2^{k-1}+1)(2^k-1))\equiv 2^k+1~{\rm mod}~(2^{2k}-1)$, the element $F(x)$ belongs to $\mathbb{F}_{2^k}$ for any $x\in \mathbb{F}_{2^n}$. Choosing $\alpha\in \mathbb{F}_{2^{n}}$ such that $\alpha+\alpha^{2^k}=1$, then ${\rm Tr}^{n}_{k}(\alpha F(x))=F(x)$ is a vectorial bent function by \cite[Proposition 3]{Pott}. This is why $F(x)$ can have the maximal number of bent components if one had known that $F(x)$ is a vectorial bent $(n,k)$-function in advance, and in our words $F(x)$ is trivial.

 Below we will examine the case of the power $(n,n)$-functions, and give an observation on when it admits the maximal number of bent components. It turns out that if this is the case, it is trivial, i.e., it must be a vectorial bent $(n,k)$-function. Note that for any such $(n,n)$-function $F(x)$, let $S(F)=\{v\in \mathbb{F}_{2^{n}}~|~{\rm Tr}^{n}_{1}(v F(x)) \text{ is not bent} \}$, $S$ is a $k$-dimensional subspace of $\mathbb{F}_{2^{n}}$. Then by the fact that $S$ is isomorphic to $\mathbb{F}_{2^{k}}\subseteq \mathbb{F}_{2^{n}}$ as the vector spaces over $\mathbb{F}_{2}$, one can easily see that $F(x)$ is EA-equivalent to a function $G(x)$ with $S(G)=\mathbb{F}_{2^{k}}$. It indicates that one only needs to consider those functions $G$ with $S(G)=\mathbb{F}_{2^{k}}$ in the sense of EA-equivalent.

\begin{thm} Let $n=2k$, $s$ be  positive integers, and $F(x)=x^s$ an $(n,n)$-function. If ${\rm Tr}^{n}_{1}(\alpha F(x))$ is a Boolean bent function for any $\alpha \in \mathbb{F}_{2^{n}}\backslash \mathbb{F}_{2^{k}}$ ($S(F)=\mathbb{F}_{2^{k}}$), then $2^k+1~|~s$, and $d=\frac{s}{2^k+1}$ satisfying that ${\rm gcd}(d,2^{k}-1)=1$. In particular, $F(x)$ admits the maximal number of bent components if and only if it is a vectorial bent $(n,k)$-function.

\end{thm}
\begin{proof} Assume to the contrary, $2^k+1\not|s$, then there exists an element $x_{0}^s\not\in \mathbb{F}_{2^{k}}, $ where $x_{0}\in \mathbb{F}^{\ast}_{2^{n}}$. We have
\begin{eqnarray*}W_{F}(\lambda,x_{0}^s)
             &\!=\!& \sum\limits_{x\in \mathbb{F}_{2^{n}}}(-1)^{{\rm Tr}^{n}_{1}((x_{0}x)^s+\lambda x)}\\
             &\!=\!&  \sum\limits_{x\in \mathbb{F}_{2^{n}}}(-1)^{{\rm Tr}^{n}_{1}(x^s)+{\rm Tr}^{n}_{1}(\lambda x^{-1}_{0}x)}\\
             &\!=\!& W_{F}(\lambda x_{0}^{-1},1).
             \end{eqnarray*}
However, by the assumption and Theorem \ref{number}, one has ${\rm Tr}^{n}_{1}(\alpha F(x))$ is not  bent for any $\alpha\in \mathbb{F}_{2^{k}}$, and then the equation that  $W_{F}(\lambda,x_{0}^s)=W_{F}(\lambda x_{0}^{-1},1)$ must induce a contradiction because of  $x_{0}^{s}\in \mathbb{F}_{2^{n}}\backslash \mathbb{F}_{2^{k}} $, and $1\in \mathbb{F}_{2^{k}}$.

Denote $\frac{s}{2^k+1}$ by $d$, and $x^{2^{k}}$ by $\overline{x}$, we shall show ${\rm gcd}(d,2^{k}-1)=1$. Now $F(x)=x^s=x^{d(2^{k}+1)}=(x\overline{x})^{d}.$ Let $\gamma \in  \mathbb{F}_{2^{n}}\backslash \mathbb{F}_{2^{k}}$, then the elments $x  \in  \mathbb{F}_{2^{n}}$ can be written uniquely in the form $a+b\gamma$ with $a, b \in \mathbb{F}_{2^{k}}.$ We have

\begin{eqnarray*}W_{F}(\lambda,a+b\gamma)
             &\!=\!&\sum\limits_{x\in \mathbb{F}_{2^{n}}}(-1)^{{\rm Tr}^{n}_{1}((a+b\gamma)x^s)+{\rm Tr}^{n}_{1}(\lambda x)}\\
             &\!=\!&\sum\limits_{x\in \mathbb{F}_{2^{n}}}(-1)^{{\rm Tr}^{n}_{1}((ax^s+b\gamma x^s)+{\rm Tr}^{n}_{1}(\lambda x)}\\
             &\!=\!&\sum\limits_{x\in \mathbb{F}_{2^{n}}}(-1)^{{\rm Tr}^{n}_{1}(b\gamma x^s)+{\rm Tr}^{n}_{1}(\lambda x)}~(\text{recall~that}~ax^s\in \mathbb{F}_{2^{k}}).\\
             \end{eqnarray*}
Let $U=\{x\in \mathbb{F}_{2^{n}} ~|~x\overline{x}=1\}$, then for any $x\in \mathbb{F}^{\ast}_{2^{n}}$, it can be uniquely written as $x=yu$, where $y\in \mathbb{F}^{\ast}_{2^{k}}$, and $u\in U$. We have

\begin{eqnarray*}W_{F}(0,a+b\gamma)
             &\!=\!&\sum\limits_{y\in \mathbb{F}^{\ast}_{2^{k}},u\in U}(-1)^{{\rm Tr}^{n}_{1}(b\gamma (yu)^s)}+1\\
             &\!=\!&\sum\limits_{u\in U}\sum\limits_{y\in \mathbb{F}^{\ast}_{2^{k}}}(-1)^{{\rm Tr}^{n}_{1}(b\gamma y^{2d})}+1~(\text{recall~ that~} u^s=1,y^s=(y\overline{y})^d=y^{2d})\\
             &\!=\!&\sum\limits_{u\in U}\sum\limits_{y\in \mathbb{F}_{2^{k}}}(-1)^{{\rm Tr}^{n}_{1}(b\gamma y^{2d})}-2^{k}\\
             &\!=\!&\sum\limits_{u\in U}\sum\limits_{y\in \mathbb{F}_{2^{k}}}(-1)^{{\rm Tr}^{k}_{1}(b(\gamma+\overline{\gamma}))y^{2d})}-2^{k}.\\
             \end{eqnarray*}
Denote $\sum\limits_{y\in \mathbb{F}_{2^{k}}}(-1)^{{\rm Tr}^{k}_{1}(b(\gamma+\overline{\gamma}))y^{2d}}$ by $A$, for $b\neq 0$, we have
$W_{F}(0,a+b\gamma)=(2^k+1)A-2^k=2^{k}(A-1)+A=\pm 2^{k}$. It imply that $(2-A)2^{k}=A$ or $A(2^{k}+1)=0$. If the former holds, then $A\neq 2$, and by the fact $\frac{A}{2-A}=2^k$ is an integer we have $A=1$. However, it would imply that $2^{k}(2-A)=2^{k}=A=1$, a contradiction with the assumption that $k$ is a positive integer. Then $A=0$, and it means that $\sum\limits_{y\in \mathbb{F}_{2^{k}}}(-1)^{{\rm Tr}^{k}_{1}(b(\gamma+\overline{\gamma}))y^{2d}}=0$ for any $b\in \mathbb{F}^{\ast}_{2^{k}}$.
Noting that $(\gamma+\overline{\gamma})x$ is a bijective map from $\mathbb{F}^{\ast}_{2^{k}}$ to $\mathbb{F}^{\ast}_{2^{k}}$, we have for any $b\in \mathbb{F}^{\ast}_{2^{k}}$, $\sum\limits_{y\in \mathbb{F}_{2^{k}}}(-1)^{{\rm Tr}^{k}_{1}(by^{d})}=0$. Therefore, the function $y^d$ is a permutation on $\mathbb{F}_{2^{k}}$, that is, ${\rm gcd}(d,2^{k}-1)=1$. The last assertion can be seen from the fact that $F(x)\in \mathbb{F}_{2^{k}}$  for any $x\in \mathbb{F}_{2^{n}}$ and \cite[Proposition 3]{Pott}.\end{proof}

\begin{rmk} Let $n=2k$, $F(x)=x^{d(2^k+1)}$, ${\rm gcd}(d,2^k-1)=1$. It is easy to see the following two facts: 1) for any $\alpha \in \mathbb{F}_{2^{k}}$, ${\rm Tr}^{n}_{1}(\alpha F(x))$ is not bent; 2) the  Walsh spectrums are the same for ${\rm Tr}^{n}_{1}(\alpha F(x))$, ${\rm Tr}^{n}_{1}(\beta F(x))$, where  $\alpha,\beta \in \mathbb{F}_{2^{n}} \backslash \mathbb{F}_{2^{k}}$. Therefore, it suffices to consider only one component function ${\rm Tr}^{n}_{1}(\alpha F(x))$ for $\alpha \in \mathbb{F}_{2^{n}}\backslash \mathbb{F}_{2^{k}}$ when one wants to compute  the number of the bent components of $F(x)$. As it has been shown in \cite{Hell09} by Helleseth et al., if $d(2^l+1)\equiv 2^i~{\rm mod} (2^k-1)$ for some integers $l\geq 1$, $i\geq 0$~(in this case ${\rm gcd}(d,2^k-1)=1$), for any $\alpha \in \mathbb{F}_{2^{n}}\backslash \mathbb{F}_{2^{k}}$, $\{W_{F}(\lambda,\alpha)~|~ \lambda \in \mathbb{F}^{\ast}_{2^n} \}$ takes on either  exactly  three or four values depending on whether $l$  and $n$ are coprime or not. It means that for those $d$, $F(x)$ does not have any bent components. Note that if $l=0$, and $d(2^l+1)\equiv 2^i~{\rm mod} (2^k-1)$, then $d \equiv 2^j~{\rm mod} (2^k-1)$ for some $j$. In this case, $F(x)=x^{d(2^k+1)}=x^{2^{i}(2^k+1)}$ has $2^n-2^k$ bent components (Gold function), and $\{W_{F}(\lambda,\alpha)~|~ \lambda \in \mathbb{F}^{\ast}_{2^n} \}$ takes on exactly two values.  Based on those observations above, for $n>8$, they conjectured that if $\{W_{F}(\lambda,\alpha)~|~ \lambda \in \mathbb{F}^{\ast}_{2^n} \}$ takes on at most four values, then there must $l\geq0$, and $i\geq 0$ such that  $d(2^l+1)\equiv 2^i~{\rm mod} (2^k-1)$, see  the arguments above Remark 3 in \cite{Hell09} for more information.
\end{rmk}

It is an interesting and important problem to determine all the exponents $d$ such that $F(x)=x^{d(2^k+1)}$ has some bent components. Inspired by the above arguments, one may expect that $d=2^{i}$ are the only ones.

{\bf Open Problem 1.} Let $n=2k$, $d$ be a positive integer such that ${\rm gcd}(d,2^k-1)=1$. To show that if the $(n,n)$-function $F(x)=x^{d(2^k+1)}$ has some bent components, then $d \equiv 2^i~{\rm mod}~(2^k-1)$ for some $i\geq 0$.

\section{Nontrivial $(n, m)$-functions with maximal number of bent components and their properties}

It has been found that for the trivial $(n,m)$-functions with maximal number of bent components, their differential spectrum takes two values $0$ and $2^k$ only, and the Walsh spectrum takes the values $0, \pm2^k$ and $2^{2k}$ only. In this section, we are interested in the constructions of nontrivial $(n,m)$-functions with maximal number of bent components. Below we will present in bivariate form a large class of such functions, whose bent components belong to the Maiorana-McFarland class $\mathcal{M}$ of bent functions, and in univariate form a class of such functions which covers all the previously known nontrivial ones. \vspace{2mm}

\noindent {\bf A. A nontrivial construction } \vspace{2mm}

The following theorem gives a large class of nontrivial $(n,m)$-functions with maximal number of bent components.

\begin{thm}\label{keythm-1} Let $m, k$ be two positive integers such that $m \geq k$. Let $G(x,y): \mathbb{F}_{2^{k}}\times \mathbb{F}_{2^{k}}\rightarrow\mathbb{F}_{2^{k}}\times \mathbb{F}_{2^{m-k}}$ be defined by $(x,y)\mapsto (x^{2^i}\pi(y)+g(y), h(y))$, where $\pi(y)$ is a permutation over $\mathbb{F}_{2^{k}}$, and $g(y)$ and $h(y)$ are arbitrary functions from $\mathbb{F}_{2^{k}}$ to $\mathbb{F}_{2^{k}}$ and to $\mathbb{F}_{2^{m-k}}$, respectively. Then for any $(0,0)\neq (u,v) \in  \mathbb{F}_{2^{k}}\times \mathbb{F}_{2^{m-k}}$, the function $\langle (u,v),G(x,y)\rangle$ is a bent component of $G$ if and only if $u\neq 0$. In particular, the function $G$ admits the maximal number $2^{m}-2^{m-k}$ of bent components.
\end{thm}
\begin{proof} For any $(0,0)\neq (u,v) \in  \mathbb{F}_{2^{k}}\times \mathbb{F}_{2^{m-k}}$, we have $\langle (u,v), G(x,y)\rangle=\langle u,
x^{2^{i}}\pi(y)+g(y)\rangle+\langle v,h(y)\rangle={\rm Tr}^{k}_{1}(u x^{2^{i}}\pi(y))+{\rm Tr}^{k}_{1}(u g(y))+{\rm Tr}^{m-k}_{1}(vh(y)):=f(x,y)$. If $u\neq 0$,  then $f(x,y)$ is in the class $\mathcal{M}$ of bent functions. Hence $G$ has at least $(2^k-1)2^{m-k}=2^m-2^{m-k}$ bent components. While according to Theorem \ref{number}, the number of bent components is at most $2^m-2^{m-k}$. Hence $f$ is bent if and only if $u\neq 0$, and $G$ has the maximal number of bent components.
\end{proof}

Below we consider the differential spectrum and Walsh spectrum of  the functions given by Theorem \ref{keythm-1}.

\begin{thm}\label{delta-1} Let $G(x,y)=(x^{2^i}\pi(y)+g(y),h(y))$, where $\pi(y)$ is a permutation over $\mathbb{F}_{2^{k}}$, and $g(y)$ and $h(y)$ are arbitrary $(k,k)$-function and $(k, m-k)$-function respectively. Then the differential spectrum of the function $G$ is given by
\begin{eqnarray*}\delta_{G}((u,v),(a,b))\!= \! \begin{cases} 0, &\text{if ~} v\!=\!0 \text{~and~ } b\!\neq\! 0,\\
2^{k}, &\text{if ~} v\!=\!0 \text{~and~ } b\!=\!0, \\  \delta_{h}(v,b),  &\text{otherwise. }
\end{cases}
\end{eqnarray*} In particular, $G$ is differentially $2^k$-uniform.
\end{thm}
\begin{proof} We have
\begin{eqnarray*}&&G(x,y)\!+\!G(x\!+\!u,y\!+\!v)\\
&\!=\!&(x^{2^{i}}\pi(y)+g(y),h(y))\!+\!(x\!+\!u)^{2^{i}}\pi(y\!+\!v)+g(y\!+\!v),h(y\!+\!v)) \\
             &\!=\!&(x^{2^{i}}\pi(y)\!+\!(x\!+\!u)^{2^{i}}\pi(y\!+\!v)\!+\!g(y)+\!g(y\!+\!v),h(y)\!+\!h(y\!+\!v))\\
             &\!=\!&((\pi(y)\!+\!\pi(y\!+\!v))x^{2^{i}}\!+\!u^{2^{i}}\pi(y\!+\!v)+\!g(y)+\!g(y\!+\!v),h(y)\!+\!h(y\!+\!v)).
             \end{eqnarray*}

Then $G(x,y)+G(x+u,y+v)=(a,b)$ is equivalent to
\begin{align}&((\pi(y)\!+\!\pi(y\!+\!v))x^{2^{i}}\!+\!u^{2^{i}}\pi(y\!+\!v)+\!g(y)+\!g(y\!+\!v)\!=\!a,\label{eq1}\\
&h(y)\!+\!h(y\!+\!v)\!=\!b. \label{eq2}
\end{align}
If $v=0$ (and then $u\neq 0$) and $b=0$, then (\ref{eq2}) holds for any $y\in \mathbb{F}_{2^{k}}$. Note that (\ref{eq1}) is reduced to $u^{2^{i}}\pi(y)=a$, which has exactly one solution $y=\pi^{-1}(au^{-2^{i}})$. Therefore, the solutions of $G(x,y)+G(x+u,y+v)=(a,b)$  are exactly those $(x,\pi^{-1}(au^{-2^{i}}))$ with $x\in \mathbb{F}_{2^{k}}$, and thus $\delta_{G}((u,0),(a,0))=2^{k}$.

If $v=0$ and $b\neq 0$,  (\ref{eq2}) has no solutions, then $\delta_{G}((u,0),(a,b))=0$.

When $v\neq 0$, one has $\pi(y)+\pi(y+v)\neq 0$, since $\pi$ is a permutation function. Then Eq. (\ref{eq1}) can be rewritten as $x^{2^{i}}=(a+u^{2^{i}}\pi(y+v)+g(y)+g(y+v))[\pi(y)+\pi(y+v)]^{-1}$, so that all the solutions to $G(x,y)+G(x+u,y+v)=(a,b)$ are $(((a+u^{2^{i}}\pi(y+v)+g(y)+g(y+v))[\pi(y)+\pi(y+v)]^{-1})^{2^{-i}},y)$, where $y$ satisfies (\ref{eq2}) . Therefore in this case $\delta_{G}((u,v),(a,b))=\delta_{h}(v,b)$.
\end{proof}
\begin{rmk}
One can see that at least most of the functions given by  Theorem \ref{keythm-1} are {\rm CCZ}-inequivalent to those trivial ones, since they have different differential spectrum. In fact, it is easy to see that the differential spectrum of a trivial function takes two values $0$ and $2^k$ only, and their Walsh spectrum takes the values $0, \pm2^k$ and $2^{2k}$ only.
\end{rmk}

\begin{prop}\label{walsh-1} Let $G(x,y)=(x^{2^i}\pi(y)+g(y),h(y))$, where $\pi(y)$ is a permutation over $\mathbb{F}_{2^{k}}$, and $g(y)$ and $h(y)$ are arbitrary $(k,k)$-function and $(k, m-k)$-function respectively. Then for any $(a,b)\in \mathbb{F}_{2^{k}}^2, (0,0)\neq(u,v)\in \mathbb{F}_{2^{k}}\times\mathbb{F}_{2^{m-k}}$, the Walsh spectrum of $F$ satisfies
\begin{eqnarray*}W_{G}((a,b),(u,v))\!=\!
\begin{cases}\!\pm\!~2^k, &{\rm~if~}u\!\neq\! 0,\\0, &{\rm~if~}u\!=\!0, a\neq 0,\\2^{k}W_{h}(b,v), &{\rm~otherwise.}
\end{cases}
\end{eqnarray*} Moreover, $G$ is plateaued if and only if $h$ is plateaued.
\end{prop}
\begin{proof}  Let $(a,b) \in \mathbb{F}_{2^{k}}\times \mathbb{F}_{2^{k}}$. For any $(0, 0)\neq (u,v) \in  \mathbb{F}_{2^{k}}\times \mathbb{F}_{2^{m-k}}$, we have
\begin{eqnarray*}\langle (u,v), G(x,y)\rangle &\!=\!& \langle u,
x^{2^{i}}\pi(y)\!+\!g(y)\rangle\!+\!\langle v, h(y)\rangle\\
& \!=\!& {\rm Tr}^{k}_{1}(u x^{2^{i}}\pi(y))\!+\!{\rm Tr}^{k}_{1}(u g(y))\!+\!{\rm Tr}^{m-k}_{1}(vh(y)):=f(x,y).
\end{eqnarray*}
One has
$$\arr{lll}{W_{f}(a,b)&\!= \! \sum\limits_{(x,y)\in \mathbb{F}_{2^{k}}\times \mathbb{F}_{2^{k}}}(\!-\!1)^{f(x,y)\!+\!{\rm Tr}^{k}_{1}(ax\!+\!by)}  \\
             &\!= \!\sum\limits_{y\in  \mathbb{F}_{2^{k}} }(\!-\!1)^{{\rm Tr}^{m-k}_{1}(vh(y))+{\rm Tr}^{k}_{1}(ug(y))\!+\!{\rm Tr}^{k}_{1}(by)}\sum\limits_{x\in  \mathbb{F}_{2^{k}}}(\!-\!1)^{{\rm Tr}^{k}_{1}(((u\pi(y))^{2^{k-i}}\!+\!a)x)}\\
             &\!=\! 2^{k}\sum\limits_{u\pi(y)=a^{2^i}}(\!-\!1)^{{\rm Tr}^{m-k}_{1}(vh(y))\!+\!{\rm Tr}^{k}_{1}(by)+{\rm Tr}^{k}_{1}(ug(y))}. \\
             }$$
By Theorem \ref{keythm-1}, $f(x,y)$ is bent if and only if $u\neq 0$. Hence when $u\neq 0$ it holds $W_{G}((a,b),(u,v))=\pm 2^k$.

If $u=0$ and $v\neq 0$, then for $a=0$, the solutions to $u\pi(y)=a^{2^i}$ are exactly all $y$ in $\mathbb{F}_{2^{k}}$. Therefore in this case $W_{f}(a,b)=2^{k}\sum\limits_{y\in\mathbb{F}_{2^{k}}}(-1)^{{\rm Tr}^{m-k}_{1}(vh(y))+{\rm Tr}^{k}_{1}(by)}=2^{k}W_{h}(b,v).$ While for $a\neq 0$, there is no solution to $u\pi(y)=a^{2^i}$. Hence $W_{f}(a,b)=0$.
 \end{proof}

\noindent {\bf B. More on $(n,n)$-functions}\vspace{2mm}

In \cite{Pott}, Pott et al. presented in univariate form an infinite class of $(2k,2k)$-functions $x^{2^{i}}(x+x^{2^{k}})$ with the maximal number of bent components. They noticed that these functions are of the form $F(x)=x^{2^{i}}\mathcal{L}(x)$ for some linear mapping $\mathcal{L}(x)=x+x^{2^{k}}$. Then they raised an open question (Question 2 in \cite{Pott}) of determining all those linear mappings $\mathcal{L}$ over $\mathbb{F}_{2^{2k}}$ such that the function $x\mathcal{L}(x)$ has the  maximal number of bent components. In the following, we completely determine under which conditions a large subclass of such $(n,n)$-functions will have the maximum number of bent components, and thus partially answer the open question. {\color{red} }

 It should be also noted that the $(n,n)$-functions constructed by \cite{MZTZ,Pott} are all quadratic. The proofs for that those functions have the maximum number of bent components are long, and they rely on the theory of quadratic forms. It seems that to find such functions, perhaps one has to examine the quadratic functions first.  However, in this paper, we break this impression and provide a direct and much simpler proof for the maximality of the number of bent components of the functions considered in the following theorem, which cover all the nontrivial functions in \cite{MZTZ,Pott} according to the arguments of Remark \ref{cover} below.

\begin{thm}\label{keythm-2} Let $n=2k$ be a positive integer, $e$ a positive divisor of $k$. Let $F^{i}(x)=x^{2^{i}}h({\rm Tr}^{n}_{e}(x))$, where $i$ is a nonnegative integer, $h:\mathbb{F}_{2^{e}}\rightarrow \mathbb{F}_{2^{e}}$ is an arbitrary function. Then  $F^{i}$ has the maximal number of bent components if and only if $e=k$ and $\pi$ is a permutation over $\mathbb{F}_{2^{e}}$. Moreover, when $e=k$ and $h$ is a permutation over $\mathbb{F}_{2^{e}}$, for any $\alpha \in \mathbb{F}_{2^{n}}$, ${\rm Tr}^{n}_{1}(\alpha F^{i})$ is a bent component of $F^{i}$ if and only if $\alpha \not\in \mathbb{F}_{2^{k}}$. For other cases, $F^i$ does not have any bent component.
\end{thm}
\begin{proof}
For any $x\in\mathbb{F}_{2^{n}}$, write $x=g+y$ such that $y\in \mathbb{F}_{2^{k}}$ and $g\in\mathcal{G}$, where $\mathcal{G}$ is a set of representatives of cosets of $\mathbb{F}_{2^{k}}$ in $\mathbb{F}_{2^{n}}$, say $\mathbb{F}_{2^{n}}=\cup_{g\in \mathcal{G}}(g+\mathbb{F}_{2^{k}})$ and it is a disjoint union. Denote $x^{2^{k}}$ by $\overline{x}$, then we have for any $(\beta, \alpha) \in \mathbb{F}_{2^{n}}\times\mathbb{F}_{2^{n}}^*$ that $W_{F^{i}}(\beta,\alpha)$ equals
\begin{eqnarray*}
&&\sum\limits_{x\in \mathbb{F}_{2^{n}}}(-1)^{{\rm Tr}^{n}_{1}(\alpha F^{i}(x))+{\rm Tr}^{n}_{1}(\beta x)} \\
&=&\sum\limits_{x\in \mathbb{F}_{2^{n}}}(-1)^{{\rm Tr}^{k}_{1}({\rm Tr}^{n}_{k}(\alpha F^{i}(x)))+{\rm Tr}^{k}_{1}({\rm Tr}^{n}_{k}(\beta x))} \\
&=&\sum\limits_{g\in \mathcal{G}, y\in  \mathbb{F}_{2^{k}}}(-1)^{{\rm Tr}^{k}_{1}(h({\rm Tr}^{k}_{e}(g+\overline{g}))((\alpha g^{2^{i}}+\overline{\alpha g^{2^{i}}})+(\alpha + \overline{\alpha})y^{2^{i}}))+{\rm Tr}^{k}_{1}((\beta g+\overline{\beta g})+(\beta+\overline{\beta})y)}\\
&=&\sum\limits_{g\in\mathcal{G}}(-1)^{{\rm Tr}^{k}_{1}(h({\rm Tr}^{k}_{e}(g+\overline{g}))(\alpha g^{2^{i}}+\overline{\alpha g^{2^{i}}})+(\beta g+\overline{\beta g}))}\sum\limits_{y\in \mathbb{F}_{2^{k}}}(-1)^{{\rm Tr}^{k}_{1}(((\alpha+\overline{\alpha})h({\rm Tr}^{k}_{e}(g+\overline{g})))^{2^{k-i}}+(\beta+\overline{\beta}))y)}\\
&=&2^{k}\sum\limits_{(\alpha+\overline{\alpha})\pi({\rm Tr}^{k}_{e}(g+\overline{g}))=(\beta+\overline{\beta})^{2^{i}}}(-1)^{{\rm Tr}^{k}_{1}(h({\rm Tr}^{k}_{e}(g+\overline{g}))(\alpha g^{2^{i}}+\overline{\alpha g^{2^{i}}})+(\beta g+\overline{\beta g}))}.
\end{eqnarray*}

 If $\alpha+\overline{\alpha}=0$, say $\alpha \in \mathbb{F}_{2^{k}}$, let $\beta$ be such that $\beta+\overline{\beta}\neq 0$, then $W_{F^{i}}(\beta,\alpha)=0$ as $(\alpha+\overline{\alpha})h({\rm Tr}^{k}_{e}(g+\overline{g}))=(\beta+\overline{\beta})^{2^{i}}$ has no solution. It implies that $\langle \alpha ,F^{i}\rangle$ is not bent. Therefore, by Theorem \ref{number}, $F^{i}$ has the maximal number of bent components if and only if for any $\alpha\not\in \mathbb{F}_{2^{k}}$, $\langle \alpha ,F^{i}\rangle$ is bent.

Now suppose $\alpha\not\in \mathbb{F}_{2^{k}}$, then $(\alpha+\overline{\alpha})h({\rm Tr}^{k}_{e}(g+\overline{g}))=(\beta+\overline{\beta})^{2^{i}}$ is equivalent to $h({\rm Tr}^{k}_{e}(g+\overline{g}))=(\beta+\overline{\beta})^{2^{i}}(\alpha+\overline{\alpha})^{-1}$. If $e=k$ and $h$ is a permutation over $\mathbb{F}_{2^{e}}=\mathbb{F}_{2^{k}}$, then the equation above reduces to $g+\overline{g}=h^{-1}((\beta+\overline{\beta})^{2^{i}}(\alpha+\overline{\alpha})^{-1})$. However,  the mapping $g\mapsto g+\overline{g}$ is a bijection from $\mathcal{G}$ to $\mathbb{F}_{2^{k}}$, hence there is exactly one solution $g$ to the equation.  It indicates that $W_{F^{i}}(\beta,\alpha)=\pm 2^k$ for any $\beta$, and thus $\langle \alpha ,F^{i}\rangle$ is bent.

If $e<k$, then for any fixed $\alpha \not\in \mathbb{F}_{2^{k}}$, by the fact $\mathbb{F}_{2^{e}}\subsetneqq \mathbb{F}_{2^{k}}$ and $(\alpha+\overline{\alpha})^{-1}({\rm Tr}^{n}_{k}(x))^{2^{i}}$ is surjective from $\mathbb{F}_{2^{n}}$ to $\mathbb{F}_{2^{k}}$, there must exist $\beta_{0}\in \mathbb{F}_{2^{n}}$ such that $(\alpha+\overline{\alpha})^{-1}(\beta_{0}+\overline{\beta_{0}})^{2^{i}}\not\in \mathbb{F}_{2^{e}}$. However, $h$ is a map from $\mathbb{F}_{2^{e}}$ to $\mathbb{F}_{2^{e}}$, it implies that $W_{F^{i}}(\beta_{0},\alpha)=0$, and hence $\langle \alpha ,F^{i}\rangle$ is not bent. Therefore  $F^{i}$ cannot have any bent component if $e<k$.

Finally suppose $e=k$, and $h$ is not a permutation over  $\mathbb{F}_{2^{e}}$, the equation reduces to $h(g+\overline{g})=(\beta+\overline{\beta})^{2^{i}}(\alpha+\overline{\alpha})^{-1}$. However, there must exist $\gamma \in \mathbb{F}_{2^{k}}$ such that its pre-image set $h^{-1}(\gamma)$ is empty. Then by the arguments as above, for any fixed $\alpha \not\in \mathbb{F}_{2^{k}}$, there must exist $\beta\in \mathbb{F}_{2^{n}}$ such that $(\alpha+\overline{\alpha})^{-1}(\beta+\overline{\beta})^{2^{i}}=\gamma$, and  then $W_{F^{i}}(\beta,\alpha)=0$. Therefore  $F^{i}$ cannot have any bent component. \end{proof}

Let $n=2k$, Proposition 3 in \cite{Pott} says that for an $(n,n)$-function $F$, if ${\rm Tr}^{n}_{1}(\alpha F(x))$ is a bent function for any $\alpha\in \mathbb{F}_{2^{n}}\backslash \mathbb{F}_{2^{k}}$, then ${\rm Tr}^{n}_{k}(\alpha F(x))$ is a vectorial bent function for any $\alpha\in \mathbb{F}_{2^{n}}\backslash \mathbb{F}_{2^{k}}$. By this observation together with Theorem \ref{keythm-2}, one can get the following result.

\begin{coro} Let $n=2k$ be a positive integer, $i$ a nonnegative integer. Let $h$ be a permutation over $\mathbb{F}_{2^{k}}$. Then the function
          ${\rm Tr}^{n}_{k}(\alpha x^{2^{i}}h({\rm Tr}^{n}_{k}(x)))$
is vectorial bent for any $\alpha \not\in \mathbb{F}_{2^{k}}$.
\end{coro}

\begin{rmk} The function $F^{i}(x)=x^{2^{i}}h({\rm Tr}^{n}_{k}(x))$, where $h$ is a permutation over $\mathbb{F}_{2^{k}}$, can possess a high algebraic degree. This class of functions provides the first nontrivial example of nonquadratic functions in univariate form with the maximal number of bent components. Some of such functions can reach the optimal algebraic degree $k$. For example, the algebraic degree of $x({\rm Tr}^{n}_{k}(x))^{2^{k-1}-1}$ is equal to $k$. \qed
\end{rmk}

\begin{rmk}\label{cover} Let $h(x)=x$ for any $x\in \mathbb{F}_{2^{k}}$. Then $h$ is a permutation over $\mathbb{F}_{2^{k}}$, and the functions $F^{i}(x)$ are just the $(n,n)$-functions $x^{2^{i}}(x+x^{2^k})$ in \cite[Theorem 4]{Pott}. Note that for a linear map $\mathcal{L}(x)$ over $\mathbb{F}_{2^{k}}$, it is a permutation of $\mathbb{F}_{2^{k}}$ if and only if $\mathcal{L}(x)=0$ has exactly only one solution $x=0$. Let $h(x)=\mathcal{L}(x)$ be a linear permutation over $\mathbb{F}_{2^{k}}$, then the functions $F^{i}(x)$ are just the $(n,n)$-functions $H(x)$ in \cite[Theorem 9]{MZTZ} which is a main result of \cite{MZTZ}.
\end{rmk}

\begin{rmk} Theorem \ref{keythm-2} partially answers Open question 2 in \cite{Pott}. Indeed, if $\mathcal{L}(x)$ is of the form $L({\rm Tr}^{n}_{k}(x))$ with $L(x)$ a linear map over $\mathbb{F}_{2^{k}}$, then $F(x)=x\mathcal{L}(x)=xL({\rm Tr}^{n}_{k}(x))$ has the maximal number of bent components if and only if $L(x)$ is a permutation over $\mathbb{F}_{2^{k}}$. It is an interesting problem to find all such functions $\mathcal{L}$, and we leave it as another open problem here.
\end{rmk}

\section{Conclusions}

We specified the bent components of an infinite class of $(n,m)$-functions in bivariate form $G(x,y)=(x^{2^{i}}\pi(y)+g(y),h(y)): \mathbb{F}_{2^{k}}\times \mathbb{F}_{2^{k}}\rightarrow\mathbb{F}_{2^{k}}\times \mathbb{F}_{2^{m-k}}$, here $\pi(y)$ is a permutation over $\mathbb{F}_{2^{k}}$, $g(y)$ is an arbitrary function from $\mathbb{F}_{2^{k}}$ to $\mathbb{F}_{2^{k}}$, and $h(y)$ is an arbitrary function from $\mathbb{F}_{2^{k}}$ to $\mathbb{F}_{2^{m-k}}$. This is the first time a class of nontrivial $(n,m)$-functions in bivariate form with maximum number of bent components were constructed.

For the $(n,n)$-function, $n=2k$, we consider two types of those functions. One is the power $(n,n)$-function $F(x)=x^s$, we showed that it admits the maximal number of bent components if and only if it is a vectorial bent $(n,k)$-function under a mild condition. The other one is of the form $F^{i}(x)=x^{2^{i}}h({\rm Tr}^{n}_{e}(x))$, where $e$ is any divisor of $k$ and $h: \mathbb{F}_{2^{e}}\rightarrow \mathbb{F}_{2^{e}}$. We showed that the functions $F^{i}(x)$ have maximum bent components if and only if $e=k$, and $h$ is a permutation over $\mathbb{F}_{2^{e}}$. An infinite class of vectorial bent functions of the form ${\rm Tr}^{n}_{k}(\alpha F^{i}(x))$ was specified, where $\alpha$ is a primitive element of $\mathbb{F}_{2^{n}}$.

%
%

\end{document}